\documentclass{article}
\usepackage{amsthm}
\usepackage{graphicx,pstricks}
\usepackage{moreverb}
\usepackage{epsfig}
\usepackage{setspace}
\usepackage{fullpage}
\usepackage{enumerate}
\usepackage{amsfonts}
\usepackage{amssymb}
\usepackage{graphicx}
\usepackage{amsmath}
\usepackage{array}
\usepackage{url}

\newcommand{\F}[0]{\mathcal{F}}
\newcommand{\C}[0]{\mathcal{C}}
\newcommand{\Fin}[0]{\mathcal{F}_\text{in}}

\newcommand{\G}[0]{\mathcal{G}}
\newcommand{\Hh}[0]{\mathcal{H}}
\newcommand{\A}[0]{\mathcal{A}}
\newcommand{\V}[0]{\mathcal{V}}
\newcommand{\K}[0]{\mathbb{K}}

\newcommand{\bs}[0]{\backslash}
\newcommand{\ep}[0]{\epsilon}

\newtheorem{thm}{Theorem}
\newtheorem{lem}[thm]{Lemma}

\newtheorem{defn}[thm]{Definition}

\begin{document}
%



\title{\textbf{On the robust hardness of \\
Gr\"{o}bner basis computation}}

%
%
%
%
%


%
\author{
%
%
\textbf{David Rolnick}\\
       Massachusetts Institute~of Technology\\
       Cambridge, MA\\
       \url{drolnick@mit.edu}\\
\and
\textbf{Gwen Spencer}\\
       Smith College\\
       Northampton, MA\\
       \url{gwenspencer@gmail.com}
}
\date{}

\maketitle
\begin{abstract}
The computation of Gr\"obner bases is an established hard problem. By contrast with many other problems, however, there has been little investigation of whether this hardness is robust.  In this paper, we frame and present results on the problem of approximate computation of Gr\"obner bases. We show that it is NP-hard to construct a Gr\"obner basis of the ideal generated by a set of polynomials, even when the algorithm is allowed to discard a $1 - \ep$ fraction of the generators, and likewise when the algorithm is allowed to discard variables (and the generators containing them). Our results shows that computation of Gr\"obner bases is robustly hard even for simple polynomial systems (e.g.~maximum degree 2, with at most 3 variables per generator). We conclude by greatly strengthening results for the Strong $c$\nobreakdash-Partial Gr\"obner problem posed by De Loera et al.~\cite{deloera}. Our proofs also establish interesting connections between the robust hardness of Gr\"obner bases and that of SAT variants and graph-coloring.
\end{abstract}





\section{Introduction}
A classic problem in computational algebra is producing a Gr\"{o}bner basis for an ideal. Given a field $\K$ and a system of polynomials $\F\subseteq \K[x_1,\ldots,x_n]$, let $\langle \F\rangle$ denote the ideal generated by the polynomials of $\F$. For a given term order, a \emph{Gr\"{o}bner basis} for the ideal $\langle \F\rangle$ is a set of generators for which the leading terms generate the ideal of leading terms in $\langle \F\rangle$.

Given a Gr\"{o}bner basis for an ideal, important problems such as ideal membership can easily be resolved; for more background, see Cox, Little, and O'Shea \cite{cox}.
The existence of a polynomial-time algorithm for computing Gr\"{o}bner bases is precluded by a number of hardness results. For general $\F$, the problem is EXPSPACE-complete \cite{kuehnle, mayr:survey}. The maximum degree of polynomials in a Gr\"{o}bner basis can also grow exponentially, even for zero-dimensional ideals \cite{giusti,mayr,moller,ritscher}.

Despite strong hardness results for the general problem of computing a Gr\"{o}bner basis, there is hope for efficient methods for restricted classes of polynomial systems. De Loera et al.~\cite{deloera} give a polynomial-time algorithm for Gr\"{o}bner basis computation in the case where $\F$ encodes the problem of properly coloring a chordal graph. Recent work of Cifuentes and Parrilo \cite{parrilo} in computational algebraic geometry has begun to unearth rich connections between restrictions on the graphical representation of the structure of a system of polynomials and efficient algorithmic methods for Gr\"{o}bner basis computation. For a related problem about computing Nullstellensatz certificates, Faug\^ere et al.~\cite{Faugere} have shown that for certain quadratic \textit{fewnomial} systems with empty varieties, imposing an extra inequality constraining the matching number of a graphical representation of the system allows polynomial-time computation of a certificate of inconsistency (and furthermore, such a certificate has linear size, while in the general case the best size bounds are exponential).

Where is the boundary between polynomial systems for which the highly-useful tool of a Gr\"{o}bner basis can be efficiently produced for $\langle \F\rangle$, and those systems for which this is impossible? Can weaker restrictions on $\F$ than those studied in \cite{parrilo} and \cite{deloera} still guarantee existence of an efficient method for Gr\"{o}bner basis computation? Conversely, can even $\F$ with quite simple representations generate ideals whose Gr\"{o}bner bases  evade efficient computation?

Instead of asking how much time and space are required to solve a problem exactly, we can ask whether the problem can be solved approximately in limited time and space.  That is, is the problem is ``robustly hard,'' resisting even approximate solution? In combinatorial optimization, the hardness of obtaining approximate solutions has been studied extensively for problems such as Graph Coloring, Traveling Salesman, Satisfiability, Maximum Cut, and Steiner Tree. For Gr\"{o}bner basis computation, however, much study has been devoted to exact hardness, but apparently little attention has been devoted to robust hardness.

Given a polynomial system $\F$ for which we would like to compute a Gr\"{o}bner basis for  $\langle\F\rangle$, what does it mean to give an approximate answer? One natural answer is to allow $\F$ itself to be approximated. In this paper, we show that even selectively throwing out a constant fraction of the generators does not allow us to compute a Gr\"{o}bner basis for the ideal corresponding to the remaining generators in polynomial time.

In the following definitions, we assume that we are given as input:  a set of polynomials $\F$ within the ring $\K[x_1,x_2,\ldots,x_n]$ and a value $\ep\in (0,1]$. We let $X=\{x_1,x_2,\ldots,x_n\}$ denote the set of all variables in our polynomial ring.
\begin{defn}[\textbf{$\ep$\nobreakdash-Fractional Gr\"obner problem}]
\label{def:fractional}
Given the input above, output:
\begin{itemize}
\item A subset $\F'\subseteq \F$ such that $|\F'| \ge \ep\cdot |\F|$.
\item A Gr\"{o}bner basis for the ideal $\langle\F'\rangle$ with respect to some lexicographic order.
\end{itemize}
\end{defn}

\begin{defn}[\textbf{Restricted $\ep$\nobreakdash-Fractional Gr\"{o}bner problem}]
\label{def:restricted}
Given the input above, output:
\begin{itemize}
\item Subsets $X'\subseteq X$ and $\F'\subseteq \F$ such that (i) $|\F'| \ge \ep\cdot |\F|$, and (ii) $\F'$ is exactly those elements of $\F$ containing only variables within $X'$.
\item A Gr\"{o}bner basis for the ideal $\langle\F'\rangle$ with respect to some lexicographic order.
\end{itemize}
\end{defn}

\noindent Note that for $\ep=1$, both problems reduce to finding a Gr\"{o}bner basis for $\langle\F\rangle$ itself.

In Section \ref{sec:fractional}, we prove that the $\ep$\nobreakdash-Fractional Gr\"obner problem is NP-hard for every $\ep > 0$ if we are working over an infinite field.  Moreover, this statement for an infinite field holds true even when ``polynomial-time'' is considered to include a dependence on the size of the Gr\"obner basis to be outputted. This is notable since in general the time to write down a Gr\"obner basis explicitly can be doubly exponential in the size of the description of the ideal.  We also show that the $\ep$\nobreakdash-Fractional Gr\"{o}bner problem is NP-hard over any field for $\ep > 7/8$.  Both these results hold even for polynomials with maximum degree 3.


\begin{thm}\label{thm:main1}
For infinite fields $\K$, the $\ep$\nobreakdash-Fractional Gr\"obner problem is NP-hard for every fixed $\ep>0$. This still holds true even if we require only an algorithm that runs in time polynomial in the larger of\\
\hspace*{4mm}(i) The size of input polynomial system $\F$,\\
\hspace*{4mm}(ii) The size of the Gr\"{o}bner basis output for $\langle\F'\rangle$. \\ 
Further, in this particular theorem the condition of a lexicographic order may be omitted. Also, this statement holds true (regardless of term order) even when each polynomial from $\mathcal{F}$  has maximum degree 3.  \end{thm}

\begin{thm}\label{thm:main2}
The $\ep$\nobreakdash-Fractional Gr\"obner problem is NP-hard for any $\ep > 7/8$ (for any field $\K$). This statement holds true even if each element of $\F$ contains at most 3 variables and even when $\mathcal{F}$ has maximum degree 3.
\end{thm}


In Section \ref{sec:restricted} we show that the Restricted $\ep$\nobreakdash-Fractional Gr\"obner problem is NP-hard for every $\ep>7/10$, assuming the field does not have characteristic 2. This result holds even if polynomials are constrained to have at most 3 variables and degree at most 3. A similar result holds for $\ep>4/5$ for any field and with polynomials constrained to have degree at most 2.


\begin{thm}
\label{extfrachard}
Assuming $\K\ne \mathbb{F}_2$, the Restricted $\ep$\nobreakdash-Fractional Gr\"{o}bner Basis Problem is NP-Hard for any $\ep>7/10$. This statement holds true even when each polynomial from $\mathcal{F}$ contains at most 3 variables, and even when $\mathcal{F}$ has maximum degree 3.
\end{thm}


\begin{thm}
\label{extfracharddeg2}
The Restricted $\ep$\nobreakdash-Fractional Gr\"{o}bner Basis Problem is NP-Hard for any $\ep>4/5$ (and for any field $\K$). This statement holds true even when each polynomial from $\mathcal{F}$ contains at most 3 variables, and even when $\mathcal{F}$ has maximum degree 2.
\end{thm}

Finally, in Section \ref{sec:partial} we recall the notions of approximate hardness posed in De Loera, et al.~\cite{deloera}.  The authors define the \textit{Weak $c$\nobreakdash-Partial Gr\"obner} problem, in which the algorithm may ignore any $c$ variables (and all generators containing them), and the \textit{Strong $c$\nobreakdash-Partial Gr\"{o}bner} problem in which the algorithm may ignore $c$ independent sets of variables, where an \emph{independent set} of variables is a set such that no pair co-occur in any generator.

\begin{defn}
Given a set $\F$ of polynomials on a set $X$ of variables, we say that $X'\subset X$ is an \emph{independent set}  of variables if no two variables from $X'$ appear in a single element of $\F$.
\end{defn}

\begin{defn}[\textbf{Strong $c$\nobreakdash-Partial Gr\"{o}bner problem} - De Loera et al.~\cite{deloera}]
Given as input, a set $\F$ of polynomials on a set $X$ of variables and a parameter $c$, output the following:
\begin{itemize}
\item disjoint $X_1,\ldots,X_b\subseteq X$, such that $b\leq c$ and each $X_i$ is an independent set of variables,
\item a Gr\"{o}bner basis for $\langle \F_{in}\rangle$ over $X_{in}=X\bs (\cup_{i=1}^b X_i)$ (i.e., we have taken away at most $c$ independent sets of variables), where the monomial order of $X$ is restricted to a monomial order on $X_{in}$.
\end{itemize}
\end{defn}

De Loera et al.~\cite{deloera} showed that for any system of polynomials subject to degree bound $k\geq 3$, the Strong $(2\lfloor\frac{k}{3}\rfloor-1)$\nobreakdash-Partial Gr\"{o}bner problem is NP-hard. That is, even giving an algorithm the freedom to choose to ignore $(2\lfloor\frac{k}{3}\rfloor-1)$ independent sets of variables (and all of the polynomials in which they appear) doesn't necessarily yield a polynomial system with an ideal for which Gr\"{o}bner basis computation is possible in polynomial time.





We show that a rather stronger statement may be proved with a simpler (though less interesting) construction, in which we follow the proof of De Loera et al.~\cite{deloera} showing NP-hardness for the Weak $c$\nobreakdash-Partial Gr\"obner problem. Further, while the result in ~\cite{deloera} required a lexicographic term order, the following result holds true for any polynomial term order.

\begin{thm} 
The Strong $c$\nobreakdash-Partial Gr\"obner problem is NP-hard for every $c$, even for polynomials of degree 3.  This still holds true even if we require only an algorithm that runs in time polynomial in the \emph{size of the Gr\"{o}bner basis} it outputs.
\label{thm:stronger}
\end{thm}

The reduction used to prove this result is perhaps not as instructive as that in De Loera et al.~\cite{deloera}. We therefore also provide an adaptation of the original method. Our statement below gives a stronger parameter for $c$ than in De Loera et al.~and furthermore does not require lexicographic order.

\begin{thm}\label{prop:lundgroebner}
For every constant $h>0$, there is a constant $k_h$ such that the following problem is NP-hard: Given a polynomial system of maximum degree $k_h$, solve the Strong $(hk_h-1)$\nobreakdash-Partial Gr\"{o}bner problem for some term order. This still holds true even if each generator in $\F$ contains at most 2 variables.
\end{thm}

The more interesting construction which proves Theorem \ref{prop:lundgroebner} can be adapted in a straightforward way to cast results very much like Theorem 9 as consequences of recent (unresolved) conjectures related to the Unique Games Conjecture \cite{Khot02}.

Our theorems apply even for polynomial systems that belong to a highly restricted classes.  In the case of Theorem \ref{thm:main2}, none of the generators contains more than 3 variables, and all coefficients come from the set $\{0,-1\}$. Qualitatively, our robust hardness results hold true even for very simple $\F$. In fact, we can even gain another significant restriction on $\F$ without compromising our robust hardness result.

\begin{thm}[Extension of Theorem \ref{thm:main2}] \label{thm:main3}
The $\ep$\nobreakdash-Fractional Gr\"{o}bner problem is NP-hard for every $\ep > 7/8$, even when the input contains only polynomials of the form $x_ix_jx_k$ or $(x_i-1)(x_j-1)(x_k-1)$.
\end{thm}

\noindent\textbf{A remark on graphical representations of $\F$ and $\F'$.}
Graphical representations of polynomial systems have been previously studied in connection with efficient Gr\"obner basis computation  (see \cite{parrilo} and \cite{deloera}). 
Given a polynomial system $\F$, construct a multigraph $G_{\F}$ as follows. For each variable $x_i$ create a node $n_i$. For each polynomial, $f_i\in \F$, create a clique on the nodes corresponding to variables that appear in $f_i$.
It has been shown that when $G_{\F}$ satisfies very special conditions, a Gr\"obner basis for $\langle\F\rangle$ can be efficiently computed. Could it be that a much wider class of polynomial systems with simple graphical representations generate ideals that admit efficient Gr\"obner basis computation? 

Our results suggest insight in the negative direction. In our Theorems \ref{thm:main1}, \ref{thm:main2}, \ref{extfrachard}, and \ref{extfracharddeg2}, the graphical representation of the input $\F$ (prior to discarding any polynomials) is already remarkably simple: each polynomial of $\F$ gives rise to a single triangle in $G_{\F}$. That is, $G_{\F}$ is the union of a ``small number'' of triangles.
Discarding $\F\backslash \F'$ produces a subgraph, $G_{\F'}$. For our unrestricted problem (Definition \ref{def:fractional}), our theorems describe the ability to select arbitrary triangles for removal from the graphical representation of $\F$. For our restricted problem (Definition \ref{def:restricted}), choosing variables to ignore corresponds to node removals that force the removal of certain triangles from $G_{\F}$. Qualitatively, our results say that any efficient method for choosing a constant fraction of the triangles to remove from (an already quite simple) $G_{\F}$ will not guarantee efficient computation of a Gr\"{o}bner basis for the ideal generated by the remaining subsystem with graphical representation $G_{\F'}$.\\



\section{Proofs for the $\ep$\nobreakdash-Fractional Gr\"obner problem.}
\label{sec:fractional}

Our proofs of Theorems \ref{thm:main1} and \ref{thm:main2} are by reduction from one of the most famous problems in combinatorial optimization, 3SAT. Results are known about the hardness of producing approximate solutions for many variants of this problem, even when the inputs are highly restricted. Such results play a key role in our arguments in this section. Recall the following problem statements:

\begin{defn}
Let $\Phi$ be a Boolean formula over a set of variables $\{y_1,\ldots,y_n\}$.  Suppose that $\Phi$ is in 3\nobreakdash-conjunctive normal form (3\nobreakdash-CNF). In other words, $\Phi$ is the conjunction of a set $\C$ of clauses, where each clause is the disjunction of three literals (of the form either $y_i$ or $\neg y_i$).  

\begin{itemize}
\item \textbf{3\nobreakdash-Satisfiability} (3SAT): Determine whether there exists a truth assignment to $\{y_i\}$ satisfying the formula $\Phi$.

\item \textbf{Maximum 3\nobreakdash-Satisfiability} (MAX-3SAT): Output a truth assignment for $\{y_i\}$ that satisfies the maximum fraction of the clauses in $\C$.
\end{itemize}
\end{defn}


\begin{proof}[Proof of Theorem \ref{thm:main1}]
We prove the theorem by reduction from 3SAT.  Let $\Phi$ be a 3\nobreakdash-CNF Boolean formula in variables $y_1,\ldots,y_n$, with clauses $C_1, \ldots, C_m$. For each $C_i$, define a polynomial $f_i(x_1,\ldots,x_n)$ as the product of three terms chosen as follows: For each positive literal $y_j$ in $C$, include the term $x_j-1$; and for each negative literal $\neg y_j$, include the term $x_j$. Thus, for example, the clause $C_i = y_3\vee \neg y_6\vee y_{11}$ gives us the function $f_i = (x_3 - 1)(x_6)(x_{11} - 1)$. This gives a set of $m$ polynomials of maximum degree 3.

Observe that truth assignments satisfying $\Phi$ correspond to simultaneous roots of the system $\{f_i\}$, where $y_j = \texttt{True}$ corresponds to $x_j = 1$, and $y_j = \texttt{False}$ to $x_j = 0$. In particular, this means that $\Phi$ is satisfiable if and only if the set $\{f_i\}$ defines a nontrivial variety.

Pick $M \ge m/\ep$, and pick $M$ distinct elements $a_1,\ldots,a_M \in \K$.  Construct the $M \times m$ matrix:

$$A = \left[\begin{array}{ccccc}
1 & a_1 & a_1^2 & \cdots & a_1^{m-1} \\
1 & a_2 & a_2^2 & \cdots & a_2^{m-1} \\
\vdots & \vdots & \vdots & \ddots & \vdots \\
1 & a_M & a_M^2 & \cdots & a_M^{m-1}
\end{array}\right].$$

Observe that any $m$ distinct rows $k_1,\ldots,k_m$ of $A$ form a Vandermonde matrix $A'$, which has determinant $\prod_{1\le r, s\le n} (a_{k_r} - a_{k_s})$.  This determinant is nonzero by our choice of $a_1, \ldots, a_m$, and hence is $A'$ invertible.

Let $\textbf{f}$ denote the vector $(f_i)$, and define the polynomials $g_1, \ldots, g_M$ as the entries of the vector $A\cdot \textbf{f}$.  For $S$ a subset of $\{1,2,\ldots, M\}$, let $A_S$ denote the $|S|\times m$ submatrix of $A$ obtained by taking exactly those rows corresponding to elements of $S$, and let $I_S$ be the ideal generated by $\{g_k\}_{k\in S}$.

Suppose that $S$ has cardinality at least $m$. Then, since each $m\times m$ submatrix of $A_S$ is invertible, the only solution to the vector equality $A\cdot \textbf{f} = \textbf{0}$ is when the vector $\textbf{f}$ is identically zero.  By the definition of $g_k$, then, simultaneous roots of $\{g_k\}_{k\in S}$ must correspond to simultaneous roots of $\{f_1,\ldots,f_m\}$.  Thus, $I_S$ is nontrivial exactly when $\{f_i\}$ defines a nontrivial variety, and hence exactly when $\Phi$ is satisfiable.  To complete the argument, note that $I_S$ is trivial exactly when its Gr\"obner basis is $\{1\}$.

Since this argument holds for every subset $S$ with $|S|\ge \ep\cdot M \ge m$, we conclude that solving the $\ep$\nobreakdash-Fractional Gr\"obner Basis Problem for input system $\{g_i\}$ permits a polynomial-time reduction\footnote{Notice that bit-wise representations of entries in $A$ have polynomial size.} for 3SAT and hence is NP-hard.  Moreover, our reduction uses 
degree-3 generators and merely requires us to determine whether some large-enough subset of the generators from $\{g_i\}$ has Gr\"obner basis $\{1\}$ (or not). Therefore, if there existed an algorithm $\mathcal{A}$ for the $\ep$\nobreakdash-Fractional Gr\"obner problem that ran in time polynomial in the size of Gr\"obner basis to be outputted, then we could simply run $\mathcal{A}$ until the time required for the output $\{1\}$.  If $\mathcal{A}$ did not terminate by this time or outputted a different Gr\"obner basis, then we would conclude that $\Phi$ has no solution.  This shows that such an algorithm $\mathcal{A}$ is possible only if P=NP. 
\end{proof}


We prove Theorem \ref{thm:main2} by reduction from MAX-3SAT, using H{\aa}stad's celebrated hardness-of-approximation result for satisfiable instances \cite{hast01}.  We use the fact that under a lexicographic order, possession of a Gr\"obner basis  will allow efficient computation of a point in the variety of the corresponding polynomial system by iterative elimination of one variable at a time.

\begin{defn}
An algorithm $\A$ is called a \emph{$\beta$\nobreakdash-approximation algorithm} for an optimization problem if $\A$ returns a solution whose value is within a $\beta$-multiplicative factor of the optimal value.
\end{defn}

\begin{thm}[Theorem 6.5 in H{\aa}stad \cite{hast01}] \label{hastad}
Assuming that P$\neq$NP, and $\delta>0$, there is no polynomial-time $(7/8+\delta)$-approximation algorithm for MAX-3SAT.  This holds even when instances are guaranteed to be satisfiable - that is, when \emph{all} clauses of the Boolean formula can be satisfied.
\end{thm}

\begin{proof}[Proof of Theorem \ref{thm:main2}]
Assume towards contradiction that $\A$ is some polynomial-time algorithm that solves the $\ep$\nobreakdash-Fractional Gr\"obner problem for some $\ep > 7/8$.  Let $\Phi$ be a satisfiable 3\nobreakdash-CNF Boolean formula in variables $y_1,\ldots,y_n$, with clauses $\C = \{C_1, \ldots, C_m\}$. For each $C_i$, define a polynomial $f_i(x_1,\ldots,x_n)$ as in the proof of Theorem \ref{thm:main1}. As above, $\Phi$ is satisfiable if and only if the set $\F = \{f_i\}$ defines a nontrivial variety in the variables $X = \{x_1, \ldots, x_n\}$. We now run $\A$ on $\F$. Let $\Fin$ denote the set of generators that $\A$ picks from within $\F$.  Thus, $|\Fin| \ge \ep \cdot |\F|$. The algorithm $\A$ must compute a Gr\"{o}bner basis $\G$ for the ideal generated by the polynomials in $\Fin$.

We claim that from $\G$ we can reconstruct a solution in the variety of $\Fin$.  By the Elimination Theorem (Theorem 2 in \S3.1 of \cite{cox}), the set $\G_k = \G\cap \K[x_k, x_{k+1}, \ldots, x_n]$ is a Gr\"obner basis for $\V(\Fin) \cap \K[x_k, x_{k+1}, \ldots, x_n]$. Given $a\in \V(\Fin)$, let $\pi_k(a)$ be the element of $\K[x_k, x_{k+1}, \ldots, x_n]$ formed by projecting $a$. By the Closure Theorem (Theorem 3 in \S3.2 of \cite{cox}), the variety of solutions to $\G_k$ is the smallest affine variety containing $\pi_k(a)$ for every $a\in \V(\Fin)$. Note that by our construction of $\F$, the variety $\V(\Fin)$ is the union of a finite number of sets $V_j$, where each $V_j$ is the product of sets of the form $\{1\}$ or $\{0\}$ or the entire completion of $\K$.  This is because, if we fix all the variables but one, the remaining variable is either constrained to $1$ or $0$ or is completely unconstrained.

Hence, every projection of $\V(\Fin)$ is itself an affine variety; thus, every solution to $\G_k$ extends to a solution to $\G$. Therefore, in order to reconstruct a solution $a\in \V(\Fin)$, we can iteratively extend a partial solution in $\V(\G_k)$. Specifically, we pick some value of $x_k$ that solves $\G_k$, given that we have already picked values for $x_{k+1}, \ldots, x_n$.

This solution $a$ is a vector of length $n$ which is a mutual zero of all polynomials in $\Fin$. Each entry in the vector corresponds to some variable in our original MAX\nobreakdash-3SAT instance: if $a_j = 1$, we assign the corresponding variable $y_j$ to be true. If $a_j = 0$, we assign $y_j$ to be false. For values of $a_j$ other than $0$ or $1$ an arbitrary assignment is made for $y_j$. By construction, the fact that $a$ causes every polynomial in $\Fin$ to evaluate to 0 means that every corresponding clause evaluates as true under the truth assignment constructed so far, so that the number of true clauses is at least $|\Fin| \ge \ep\cdot |\F| = \ep \cdot |\C|$.

Thus, our constructed truth assignment is guaranteed to satisfy more clauses than Theorem \ref{hastad} permits.  Moreover, our reduction uses generators that have maximum degree 3 and contain at most 3 variables per generator (so that $\F$ is relatively sparse).  Since the reduction we have described is clearly polynomial-time (assuming that $\A$ is itself polynomial-time), we obtain a contradiction and conclude that the $\ep$\nobreakdash-Fractional Gr\"obner problem is NP-hard, as desired. 
\end{proof}

\noindent\textbf{An extension for very simple $\F$: Theorem \ref{thm:main3}}. 
We can, in fact, gain a significant additional restriction on $\F$ (without losing anything in our robust hardness result) by using a more recent inapproximability result of Guraswami and Khot \cite{guruswamikhot} for a specialized variant of MAX\nobreakdash-3SAT known as ``Max Non-Mixed Exactly 3SAT." In this variant, each clause must contain exactly 3 literals, and each clause either has all three literals in positive form or all three literals in negated form (there is ``no mixing" of positive-form and negated-form literals within clauses). Guraswami and Khot match H{\aa}stad's result for the general case: even this specialized variant is NP-hard to approximate within multiplicative factor better than $7/8$, and this is true even for instances guaranteed to be satisfiable. Our reduction works in the same way as before, but the polynomial system constructed from the arbitrary Max Non-Mixed Exactly 3SAT instance has an even simpler form. Thus, for $\ep > 7/8$, the $\ep$\nobreakdash-Fractional Gr\"{o}bner problem is NP-hard even for this more restricted class of polynomial systems.

Further,  shifting to the Restricted $\epsilon$\nobreakdash-Fractional Gr\"{o}bner Model, even for the simple clause form in Theorem \ref{thm:main3}, the result can be improved to $\ep > 3/4$. Since this parameter is weaker than our parameter in Theorem \ref{extfrachard}, we omit an explicit proof. 

\section{Proofs for the Restricted $\ep$\nobreakdash-Fractional Gr\"{o}bner problem}
\label{sec:restricted}

Compared with the $\ep$\nobreakdash-Fractional Gr\"{o}bner problem, our definiton of the  Restricted $\ep$\nobreakdash-Fractional Gr\"{o}bner problem  imposes additional structure on the set of generators that the algorithm can ignore from $\F$. In this section, we give proofs of Theorems 5 and 6 that exploit this additional constraint on structure to give reductions from other logical satisfiability problems (where the hardness of approximation bounds are lower than for MAX-3SAT). 
Reductions in this section will be more involved because the form of the polynomials that encode \textit{Not-2} and \textit{OXR} clauses don't naturally force the variety of the constructed polynomial system to be contained in $\{0,1\}^n$ (as we had for disjunctions in Section \ref{sec:fractional}). 


For a polynomial system $\F$ and subset of variables $Y$, let $\F_Y$ denote the subset of polynomials from $\F$ which contain at least one variable from $Y$.
Referring to Definition \ref{def:restricted}, we write  $X''=X\backslash X'$ to correspond to a set of variables chosen by the algorithm to be ignored. The set of polynomials containing at least one variable from $X''$, denoted by $\F_{X''}$, is ignored, and the algorithm need only compute a Gr\"{o}bner basis for the remaining set of polynomials $\mathcal{F}'$ (the set of polynomials that each contain only \textit{retained variables} from $X'$). 


We prove Theorem \ref{extfrachard} by reducing from the \textit{Max-Not-2 Problem for satisfiable instances of arity 3}. The input to the Max-Not-2 Problem is a set of logical predicates $\mathcal{P}$ over a set of literals $\mathcal{L}$. Specifying \textit{arity 3} means that each predicate contains at most 3 signed literals (a ``signed literal" is a literal in either negated or positive form), e.g. $(l_i,l_j,\neg l_k)$ where $l_i, l_j,l_k\in \mathcal{L}$. For a truth assignment to the literals, a predicate is ``satisfied" if the number of its signed literals that are true is not exactly 2. If exactly 2 of its signed literals are true, then the predicate is not satisfied.\footnote{For example, the predicate $(l_i,l_j,\neg l_k)$ is satisfied by a truth assignment where $l_i$ is true, $l_j$ is true, and $l_k $ is false (all three of the signed literals in the predicate are true for this truth assignment). On the other hand, consider a truth assignment in which $l_i$ is true, $l_j$ is true, and $l_k$ is true: for this truth assignment exactly 2 of the signed literals in the predicate are true, so the predicate is not satisfied.} The objective is to compute a truth assignment that satisfies the highest possible fraction of predicates in $\mathcal{P}$. When we study the problem \textit{for satisfiable instances} we are guaranteed that some truth assignment for $\mathcal{L}$ satisfies every predicate in $\mathcal{P}$.

H{\aa}stad recently showed that fair coin-flipping gives a tight approximation guarantee for this problem:
\begin{thm}[H{\aa}stad \cite{hast14}]\label{hastadnot2}
For any $\delta>0$, given a satisfiable instance of Max-Not-2 of arity 3, there is no polynomial-time algorithm to find a truth assignment that satisfies a $(\frac{5}{8}+\delta)$-fraction of the predicates (unless $P=NP$). 
\end{thm}
\vspace{-1mm}

Before the proof of Theorem \ref{extfrachard}, we recall the key fact our reduction will invoke from computational algebra. Given a Gr\"{o}bner basis for  $\langle\mathcal{F}\rangle$ with respect to a lexicographic order, if the variety of  $\mathcal{F}$ (the set of mutual roots of all polynomials in $\mathcal{F}$) is finite, then a point in the variety of $\mathcal{F}$ (a mutual root of all polynomials in $\mathcal{F}$) can be computed efficiently by iteratively eliminating the variables one at a time. These classic results in elimination theory are covered in the textbook of Cox, Little and O'Shea \cite{cox}.

\vspace{2mm}
\begin{proof}[Proof of Theorem \ref{extfrachard}] Assume, for  contradiction, that the $\mathcal{A}$ asserted in Theorem \ref{extfrachard} does exist with $\ep>7/10$. Define $\gamma=\ep-7/10$ so that $\gamma>0$. Given an arbitrary satisfiable input $(\mathcal{P}, \mathcal{L})$ of the \textit{Max-Not-2 Problem of arity 3} we compute a truth assignment of forbidden quality in polynomial time as follows. Our assignment will be determined over the course of three stages: instance preprocessing, polynomial encoding, and supplemental random assignment.

\vspace{2mm}

\noindent \textbf{Stage 1. Instance Preprocessing.} The form of a Not-2 predicate sometimes unequivocally dictates a literal's truth value in all satisfying assignments. For example, $(l_i,l_i,\neg l_i)$ must have $l_i$ false in every satisfying assignment. Further, if some literal appears in only one Not-2 predicate, then manipulating that literal's truth value can always satisfy the predicate (regardless of the truth values of all other literals).  

In Appendix 1, we use simple arguments like these to prove that an arbitrary satisfiable \textit{Max-Not-2} instance may be preprocessed (fixing some literals to be true/false and removing some predicates) so that
WLOG the remaining instance is still satisfiable, an $\alpha$-approximate truth assignment for the updated instance is at minimum $\alpha$-approximate for the original Max-Not-2 instance, and the following two convenient properties hold:

\vspace{1mm}

\noindent \textbf{Property 1.} \textit{Any predicate $p'\in \mathcal{P}$ that has multiple occurrences of the same literal must have a very specific form: either $p'$ contains two identical signed literals and a third signed literal corresponding to a different index, or  $p'$ contains two identical signed literals and a third literal whose truth value has been permanently fixed to either false or true.}
\vspace{1mm}

\noindent \textbf{Property 2.} \textit{After the updates in Stage 1, Every literal $l\in \mathcal{L}$ appears in some form (either negated or positive) in at least two predicates from $\mathcal{P}$.}

\vspace{1mm}

\noindent 
Thus, to obtain a contradiction, it is sufficient to show that for a preprocessed satisfiable Max-Not-2 instance (where Properties 1 and 2 hold)  we can satisfy a $(\frac{5}{8}+\epsilon)$ fraction of the predicates. First, we derive an immediate consequence of Property 2 that will be used in Stage 2. Later, in Stage 3, we will use Property 1 to analyze a final stage of supplemental random truth assignment. 
\begin{lem} \label{twofifths} An instance of \textit{Max-Not-2} of arity $3$ over literals $\mathcal{L}$ and predicates $\mathcal{P}$ for which each literal appears in at least 2 predicates has $|\mathcal{P}|\geq \frac{2}{5}(|\mathcal{P}|+|\mathcal{L}|)$.
\end{lem}

\begin{proof}[Proof of Lemma] Since there are $|\mathcal{L}|$ literals, and each appears at least twice (from Property 2), then there must be at least $2|\mathcal{L}|$ appearances of literals. Each predicate contains at most 3 appearances of literals, so at minimum there are $2|\mathcal{L}|/3$ predicates in $\mathcal{P}$:
\vspace{-4mm}
\begin{align*}
|\mathcal{P}|&\geq \frac{2}{3} |\mathcal{L}|\\
|\mathcal{P}|+\frac{2}{3}|P|&\geq \frac{2}{3} |\mathcal{L}|+\frac{2}{3}|P|\\
\frac{5}{3}|P|&\geq \frac{2}{3} (|\mathcal{L}|+|P|)\\
|\mathcal{P}|&\geq \frac{2}{5}(|\mathcal{P}|+|\mathcal{L}|).\text{\qedhere}
\end{align*}
\end{proof}

\noindent \textbf{Stage 2. Polynomial Encoding and Gr\"{o}bner-based partial assignment.} 
First we take a convenient (and equivalent) algebraic view of Max-Not-2 predicates.  Consider each literal $l_i$ to be a $\{0,1\}$ variable $x_i$ (where $x_i=1$ corresponds to $l_i$ true, and $x_i=0$ corresponds to $l_i$ false). Translate each predicate into a sum:  if a predicate contains a signed literal in positive form, a positive copy of the corresponding variable is added. If a predicate contains a signed literal in negated form, a term is added in which the corresponding variable is subtracted from 1. For example, the predicate $(l_i,l_j,\neg l_k)$ becomes the sum $x_i+x_j+(1-x_k)$. It is easy to check that the original predicate is satisfied exactly when its corresponding sum is not 2 (and hence the sum has total value 0, 1, or 3). We will say each predicate has at most three \emph{acceptable totals}, and exactly one \emph{unacceptable total}.\footnote{If the predicate has fewer than three signed literals, the number of achievable acceptable totals maybe be less than 3.}

Now, define a polynomial system based on $(\mathcal{P}, \mathcal{L})$ as follows. Create a variable $y_i$ corresponding to the $i$th literal of $\mathcal{L}$. Denote this new set of variables by $Y$. These $y_i$ will replicate the $x_i$ in the algebraic view of Max-Not-2: for each $i\in \{1,2,...,|\mathcal{L}|\}$ create a polynomial $y_i(1-y_i)$. This gives a set of $|\mathcal{L}|$ \emph{literal polynomials} whose mutual roots are exactly $\{0,1\}^{|\mathcal{L}|}$.

Next, create a polynomial corresponding to each predicate in $\mathcal{P}$. In the algebraic view of Max-Not-2, each predicate $p$ corresponds to a sum with at most 3 acceptable totals (some subset of $\{0,1,3\}$ gives the acceptable totals for which the sum corresponds to a satisfied predicate). 
Each acceptable total for the sum 
will be used to define a linear term, and the product of these linear terms will give the \emph{predicate polynomial} corresponding to $p$. Each linear term is $p$'s sum (with the fixed variables from Stage 1 substituted in, e.g. $l_i$ fixed true implies $y_i=1$) minus an acceptable total for the sum. For example, recall that the algebraic view of the Not-2 predicate $(l_i,l_j,\neg l_k)$ is the sum $x_i+x_j+(1-x_k)\neq 2$.
If none of $x_i, x_j, x_k$ were fixed in Stage 1, then this yields the following polynomial:
\[
\bigl(y_i+y_j+(1-y_k)-0\bigr))\bigl(y_i+y_j+(1-y_k)-1\bigr)\bigl(y_i+y_j+(1-y_k)-3\bigr)
\]
The first term in this product corresponds to acceptable total sum of 0, etc.
Since there are always at most 3 acceptable totals for $p$'s sum, the polynomial constructed is the product of at most 3 linear terms, and hence has degree at most 3. Since each predicate has at most 3 signed literals, each polynomial will contain at most 3 variables. When restricted to the mutual roots for the set of literal polynomials defined above, 0s for $p$'s polynomial correspond exactly to the cases in which $p$'s sum takes on an acceptable total. This encoding of the Not-2 requirement can fail only when some term in the product evaluates to 0 in a misleading way (when a sum is 2), which can only arise if $\mathbb{K}$ is $F_2$.

All properties still hold when some variables in $p$'s sum were fixed during pre-processing in Stage 1. For example, if in Stage 1, $l_i, l_j$ were not fixed but $l_k$ was fixed to false, then $x_k=0$, so $p$'s sum becomes $x_i+x_j+(1-0)\neq 2$ and consequently $p$'s polynomial construction is:
\begin{align*}
\bigl(y_i+y_j+(1-0)-0\bigr)\bigl(y_i+y_j+(1-0)-1\bigr)\bigl(y_i+y_j+(1-0)-3\bigr)=
\bigl(y_i+y_j+1\bigr)\bigl(y_i+y_j\bigr)\bigl(y_i+y_j-2\bigr).
\end{align*}

This constructs a set of predicate polynomials of size $|\mathcal{P}|$. Let $\mathcal{F}$ denote the system of polynomials containing both the literal polynomials and the predicate polynomials. Notice that every satisfying assignment for $(\mathcal{P},\mathcal{L})$ can be interpreted as a point in the variety of $\mathcal{F}$ (aka, a mutual root of all $\mathcal{F}$'s polynomials). In particular, since $(\mathcal{P},\mathcal{L})$ is satisfiable, $V(\langle \mathcal F \rangle)$ is non-empty.

Apply algorithm $\mathcal{A}$ to solve the $q$\nobreakdash-Fractional Gr\"{o}bner Problem for $\mathcal{F}$ for $\ep=(7/10+\gamma)$ for some fixed $\gamma>0$. Let $Y''$ denote the variables that $\mathcal{A}$ selects to ignore. From the definition of a Restricted $\ep$\nobreakdash-Fractional Gr\"{o}bner Basis, $Y''$ was chosen so that the set of ignored polynomials, $\mathcal{F}_{Y''}$, has bounded size:
\begin{align*}
|\mathcal{F}_{Y''}|\leq (1-\ep)|\mathcal{F}|\leq (3/10-\gamma)|\mathcal{F}|&=
(3/10-\gamma)(|\mathcal{P}|+|\mathcal{L}|)\\
&\leq (3/10-\gamma)\Big(\frac{5}{2}|\mathcal{P}|\Big)\\
&\leq \Big(\frac{3}{4}-\frac{5}{2}\gamma\Big)|\mathcal{P}|
\end{align*}
The second line follows from the first due to Lemma \ref{twofifths}. 

Let $\mathcal{P}_D$ denote the set of ignored predicate polynomials that are in $\mathcal{F}_{Y''}$, and $\mathcal{P}_R$ denote the set of retained predicate polynomials that are in $\mathcal{F}\backslash\mathcal{F}_{Y''}$. Clearly $|\mathcal{P}_D|\leq |\mathcal{F}_{Y''}|\leq \Big(\frac{3}{4}-\frac{5}{2}\gamma\Big)|\mathcal{P}|$, so:
\begin{align}\label{quarter}
|\mathcal{P}_R|= |\mathcal{P}|- |\mathcal{P}_D| \geq |\mathcal{P}|- \Big(\frac{3}{4}-\frac{5}{2}\gamma\Big)|\mathcal{P}|= \Big(\frac{1}{4}+\frac{5}{2}\gamma\Big)|\mathcal{P}|.
\end{align}That is, $\mathcal{A}$ computes a Gr\"{o}bner basis with respect to a lexicographic order for the ideal generated by the polynomials in $\mathcal{F}\backslash \mathcal{F}_{Y''}$, and inequality (\ref{quarter}) says that at least a $(\frac{1}{4}+\frac{5}{2}\gamma)$ fraction of all predicate polynomials must be in $\mathcal{F}\backslash \mathcal{F}_{Y''}$. 

The satisfiability of $(\mathcal{P}, \mathcal{L})$ ensures that a mutual root of the polynomials of $\mathcal{F}$ exists. Thus,  the polynomials from $\mathcal{F}\backslash \mathcal{F}_{Y''}$ must also have a mutual root.
Since we have a Gr\"{o}bner basis for $\langle\mathcal{F}\backslash \mathcal{F}_{Y''}\rangle$ with respect to a lexicographic order, we can solve efficiently for a mutual root of the polynomials in $\mathcal{F}\backslash \mathcal{F}_{Y''}$ via successive elimination of the variables. In particular, since the variety of $\mathcal{F}\backslash \mathcal{F}_{Y''}$ is finite (it is a subset of $\{0,1\}^{|Y\backslash Y''|})$, all partial solutions extend, and for each successive variable elimination only 2 options must be checked to find some $y_i$ that works. The resulting solution is a vector $y^*$ of length $|Y\backslash Y''|$ which is a mutual zero of all polynomials in $\mathcal{F}\backslash \mathcal{F}_{Y''}$. Each entry in the vector corresponds to some literal variable in our Max-Not-2 instance: if $y_i$ is 1 in $y^*$, assign the corresponding literal variable $x_i$ to be 1, if $y_i$ is 0 in $y^*$, assign the corresponding literal variable $x_i$ to be 0.\footnote{Because of the inclusion of the literal polynomials, and the fact that by definition $\mathcal{F}\backslash \mathcal{F}_{Y''}$ contains all the literal polynomials corresponding to variables in $Y\backslash Y''$, this routine makes an assignment for every $x_i$ corresponding to a  $y_i \in (Y\backslash Y''$).} 

The vector $y^*$ is a mutual zero of polynomials in $\mathcal{F}\backslash \mathcal{F}_{Y''}$: substituting $y^*$ into any predicate polynomial in $\mathcal{F}\backslash \mathcal{F}_{Y''}$ gives 0. By our construction of the predicate polynomials, this implies that our partial assignment for $x$ based on $y^*$ yields an acceptable total for every predicate polynomial in $\mathcal{F}\backslash \mathcal{F}_{Y''}$. That is, every predicate whose polynomial is in $\mathcal{P}_R$ is satisfied by our current partial assignment for $x$ (and from inequality (\ref{quarter}), $\mathcal{P}_R$ corresponds to strictly more than $1/4$ of the predicates from $\mathcal{P}$).

\vspace{2mm}

\noindent \textbf{Stage 3. Supplemental Random Assignment.} Let $p$ denote a
predicate corresponding to an ignored polynomial in $\mathcal{P}_D\subseteq \mathcal{F}_{Y''}$. From the form of $\mathcal{F}_{Y''}$, $p$ contains at least one literal corresponding to a $y_i\in Y''$.   
Literals corresponding to variables in $Y''$ have not yet been assigned truth values.
We exploit this to ensure that a modest fraction of predicates corresponding to polynomials from $\mathcal{P}_D$ are satisfied. 

For literals corresponding to ignored variables $y_i\in Y''$, consider an independent coin-flip procedure that assigns $x_i=1$ (a.k.a. $l_i$ true) with probability $\frac{1}{2}$, and $x_i=0$ (a.k.a. $l_i$ false) with probability $\frac{1}{2}$. 
We will argue that regardless of the partial assignment constructed for $x$ in Stage 2 (and effectively in Stage 1, through literal fixing), in expectation this random procedure satisfies at least half of the predicates corresponding to the polynomials in $\mathcal{P}_D$.

At least one of $p$'s signed literals has a truth value that is not-yet-assigned and will be decided by the coin-flipping procedure. Using the algebraic view of $p$ as a sum (e.g.$(l_i,l_j,\neg l_k)$ becomes sum $x_i+x_j+(1-x_k)$), we
analyze the probability that $p$ has an acceptable total sum at the end of the coin-flip procedure (so that $p$ is satisfied). The key point in the following case analysis is that $p$ has some \textit{current achieved sum} at the end of Stage 2 (before random coin-flipping begins), and for $p$ to be satisfied, $p$'s total sum must avoid exactly one unacceptable total (namely 2). Say that $p$'s unacceptable total minus $p$'s current achieved sum is $p$'s \textit{forbidden margin}. For example, if $p$ enters Stage 3 with no signed literals fixed, $p$'s forbidden margin will be 2. If $p$ enters Stage 3 with exactly 2 signed literals fixed to be true, and one signed literal not yet fixed, then $p$'s forbidden margin will be 0.  

First, suppose that all signed literals in $p$ correspond to unique literals. Coin flips are independent, so the possible probability spaces are as follows.

\renewcommand\arraystretch{1.2} 

\begin{tabular}{|l|l|l|} \hline
$p$'s Number of Un-fixed Signed  & $p$'s Realized Margin: Additional Signed  & Probability  \\
Literals Entering Stage 3 & Literals True After Coin-flipping& Distribution  \\
\hline
3 & $\{0,1,2,3\}$& $(\frac{1}{8},\frac{3}{8},\frac{3}{8},\frac{1}{8})$\\ 
\hline
2 & $\{0,1,2\}$& $(\frac{1}{4},\frac{1}{2},\frac{1}{4})$\\ 
\hline
1 & $\{0,1\}$& $(\frac{1}{2},\frac{1}{2})$\\ 
\hline
\end{tabular}

\vspace{2mm}

Regardless of the number of un-fixed signed literals in $p$ entering Stage 3, there is no single outcome whose probability is strictly more than $1/2$.  Thus, regardless of the value of $p$'s forbidden margin, the probability that $p$'s total sum is acceptable (that is, the realized margin avoids $p$'s forbidden margin) is at least $1/2$. Thus, if all signed literals in $p$ correspond to unique literals, the probability $p$ is satisfied is at least $1/2$.

Otherwise, the signed literals in $p$ do not correspond to unique literals. Since  Property 1 holds WLOG (as noted in Stage 1 and proved in Appendix 1), there are only two possible forms for $p$. We argue that in either case, the probability that $p$ reaches an acceptable total is at least $1/2$:

\vspace{-3mm}

\begin{enumerate}
\item Suppose $p$ contains two identical signed literals with a third signed literal that was fixed to 0 or 1 in Stage 1. Since at least one of $p$'s signed literals has not-yet-assigned truth value, 
it must be that the two identical signed literals are un-assigned entering Stage 3. If the fixed third signed literal has value 0, then the probability that $p$ avoids unacceptable total sum of 2 is $\frac{1}{2}$. If the fixed third signed literal has value 1, then the probability that $p$ avoids unacceptable total sum is $1$.

\vspace{-2mm}

\item Suppose $p$ contains two identical signed literals with a third term that is a signed literal corresponding to an unrelated index. Again $p$ has a specific forbidden margin. Consider the possible probability spaces:

\begin{tabular}{|l|l|l|} \hline
State Entering Stage 3 & $p$'s Realized Margin: Additional Signed  & Probability  \\
 & Literals True After Coin-flipping& Distribution  \\
\hline
(pair fixed at 0 or 2, single un-fixed)& $\{0,1\}$&$(\frac{1}{2},\frac{1}{2})$\\  
\hline
(pair un-fixed, single fixed at 0 or 1) & $\{0,2\}$& $(\frac{1}{2},\frac{1}{2})$\\ 
\hline
(pair un-fixed, single un-fixed) & $\{0,1,2,3\}$& $(\frac{1}{4},\frac{1}{4},\frac{1}{4},\frac{1}{4} )$\\ 
\hline
\end{tabular}

Again, regardless of $p$'s current achieved sum at the end of Stage 2, the probability of $p$'s forbidden margin being realized in Stage 3 is at most $\frac{1}{2}$. Thus, the probability of $p$ being satisfied (by reaching an acceptable total) is greater than or equal to $\frac{1}{2}$.
\end{enumerate}

\vspace{-2mm}
\noindent Thus, for arbitrary predicate $p$ corresponding to $\mathcal{P}_D$, the probability $p$ is satisfied at the end of Stage 3 is at least $ 1/2$. Since the expectation of the sum is the sum of the expectations, we have that the expected number of predicates satisfied by the coin-flipping procedure is at least $|\mathcal{P}_D|/2$.  

Such a randomized assignment procedure can then be derandomized via the well-known method of conditional expectations to obtain a deterministic assignment algorithm with quality that matches the expected guarantee.
Namely, given a list of the variables $x_j$ to be randomly assigned, proceed through the list deterministically fixing truth values of one variable at a time as follows: for each variable $x_j$, compare the conditional expected values of the number of clauses satisfied given assignments of true / false for $x_j$. One of these conditional expectations must be at least as large as the expected value when $x_j$ also is decided uniformly at random.  Thus, assigning $x_j$ the truth value with the larger conditional expectation can be used to iteratively compute a deterministic solution of quality at least as high as the original expected value of half of $|\mathcal{P}_D|$.

We finally have a full truth assignment for the literals $\mathcal{L}$ that satisfies every predicate in $\mathcal{P}_R$ and at least $1/2$ of the predicates in $\mathcal{P}_D$:

\begin{align*}
\text{Total Predicates We Satisfy}\geq|\mathcal{P}_R|+\frac{|\mathcal{P}_D|}{2} 
&= 
|\mathcal{P}_R|+\frac{(|\mathcal{P}|-|\mathcal{P}_R|)}{2}\\
&= 
\frac{|\mathcal{P}_R|}{2}+\frac{|\mathcal{P}|}{2}\\
&\geq 
\frac{1}{2}\Big(\frac{1}{4}+\frac{5}{2}\gamma\Big)|\mathcal{P}|+\frac{|\mathcal{P}|}{2}\\
&\geq 
\Big(\frac{5}{8}+\frac{5}{4}\gamma\Big)|\mathcal{P}|\\
&\geq 
\Big(\frac{5}{8}+\delta\Big)|\mathcal{P}| \hspace{5mm} \text{ for some  $\delta>0$}
\end{align*}

From the second to third line, inequality (\ref{quarter}) is applied. As $\gamma>0$, let $\delta=\frac{5}{4}\gamma$ to get the final statement about $\delta$. This reduction runs in polynomial time. The
final inequality shows that our method exceeds the hardness-of-approximation bound of H{\aa}stad for Max-Not-2 (listed as Theorem \ref{hastadnot2} earlier). Thus we have a contradiction, and Theorem \ref{extfrachard} is proved.
\end{proof}

Next we prove the first hardness result for approximate Gr\"{o}bner basis computation for polynomial systems of maximum degree 2 (matching the degree bound for NP-hardness of exact Gr\"{o}bner basis computation).  


Our proof of Theorem \ref{extfracharddeg2} is closely inspired by our proof of Theorem \ref{extfrachard}. Some notation and language introduced there will be directly reused. Differences arise from the form of the logical predicates considered: properties from preprocessing in Stage 1 are slightly different, the polynomial system constructed has lower-degree predicate polynomials, and the polynomials' form impacts our analysis of both the Gr\"{o}bner-Basis-based partial truth assignment and the final coin-flip-based portion of the truth assignment.

We rely on an earlier 2001 hardness result due to H{\aa}stad for a problem involving logical predicates of arity 3 where the system of predicates is satisfiable.  The predicates are now of the following form:

\begin{align}\label{oxrform}
OXR(q_1,q_2,q_3)=q_1 \vee (q_2 \oplus q_3)
\end{align}
Here $q_1,q_2,q_3$ are signed literals which represent positive or negated forms of literals from a set $\mathcal{L}$. This predicate is satisfied if at least one of $q_1$ or $(q_2 \oplus q_3)$ is true.  The second option $(q_2 \oplus q_3)$ is often called an ``xor" or ``exclusive or." This exclusive or is true when exactly one of $q_2$ or $q_3$ is true. Describing (\ref{oxrform}) above, we will say that $q_1$ is in the \textit{special position} of $p$ and that $q_2$ and $q_3$ are in the \textit{symmetric positions} of $p$. 

\begin{thm} [H{\aa}stad \cite{hast01}] \label{hastad2}
For any $\delta>0$, given a satisfiable instance of Max-OXR of arity 3, there is no polynomial-time algorithm to find a truth assignment that satisfies a $(\frac{6}{8}+\delta)$-fraction of the predicates (unless $P=NP$). 
\end{thm}


\begin{proof}[Proof of Theorem \ref{extfracharddeg2}]  Assume, for contradiction, that an $\mathcal{A}$ as described in Theorem \ref{extfracharddeg2} does exist with $\ep>4/5$. Define $\gamma=\ep-4/5$ so that $\gamma>0$. Given an arbitrary satisfiable input $(\mathcal{P}, \mathcal{L})$ of the \textit{Max-OXR Problem of arity 3} we compute a truth  assignment of forbidden quality in polynomial time. Similar to the proof of Theorem \ref{extfracharddeg2}, the assignment is constructed over three stages.\\

\noindent \textbf{Stage 1. Instance Preprocessing.} The form of OXR predicates sometimes unequivocally dictates a literal's truth value in all satisfying assignments. For example, if $q_1 \vee (q_3 \oplus q_3)$ is satisfied, then $q_1$ must be true. Further, if some literal appears in only one OXR predicate, then manipulating that literal's truth value can always satisfy the predicate (regardless of the position in which the literal appears or the truth values of all other literals).  

In Appendix 2,  we use simple arguments like these to prove that an arbitrary satisfiable \textit{Max-OXR} instance may be preprocessed (fixing some literals to be true/false and removing some predicates) so that
WLOG the remaining instance is still satisfiable, an $\alpha$-approximate logical assignment for the updated instance is at minimum $\alpha$-approximate for the original Max-OXR instance, and
the following two convenient properties hold:

\vspace{1mm}
\noindent \textbf{Property 1.} \textit{If $p\in \mathcal{P}$, then the two signed literals in $p$'s symmetric positions correspond to unique literals.} 
\vspace{1mm}

\noindent \textbf{Property 2.} \textit{After the updates in Stage 1, Every literal $l\in \mathcal{L}$ appears in some form (negated or positive) in at least two predicates from $\mathcal{P}$.}

\vspace{1mm}
\noindent As an immediate consequence of Property 2, we can derive Lemma \ref{twofifths} (as in our proof of Theorem \ref{extfracharddeg2}): $|\mathcal{P}|\geq \frac{2}{5}(|\mathcal{P}|+|\mathcal{L}|)$. This inequality will be used in Stage 2. Later, in Stage 3, we will use Property 1 to analyze a final round of supplemental random truth assignment.

\vspace{2mm}
\noindent \textbf{Stage 2.  Polynomial Encoding and Gr\"{o}bner-based partial assignment. } Define a polynomial system based on $(\mathcal{P}, \mathcal{L})$ as follows. Create a variable $y_i$ corresponding to the $i$th literal of $\mathcal{L}$. Denote this set of variables by $Y$. Create a polynomial $y_i(1-y_i)$. This gives a set of $|\mathcal{L}|$ literal polynomials whose mutual roots are exactly $\{0,1\}^{|\mathcal{L}|}$. As in the previous proof, $y_i=0$ implies $l_i$ is false, and $y_i=1$ implies $l_i$ is true. 

Next create a set of predicate polynomials. Consider $p\in\mathcal{P}$.
We create a polynomial corresponding to $p$ as follows: it will be the product of 2 terms.  The first term will correspond to $p$'s special position. If the signed literal in the special position of $p$ corresponds to literal $l_i$ and appears in positive form, then the first term in $p$'s polynomial will be $(y_i-1)$. If the signed literal in the special position of $p$ corresponds to literal $l_i$ and appears in negated form, then the first term in $p$'s polynomial will be $(y_i)$. The second term in $p$'s polynomial will correspond to $p$'s xor.  For $p$'s xor to be true, using language from the proof of Theorem \ref{extfrachard}, there is only one \textit{acceptable sum} for variables corresponding to $p$'s symmetric-position signed literals (namely 1). Let $l_j$ and $l_k$ denote the literals corresponding to the signed literals in the symmetric positions of $p$ ($k\neq j$ by Property 1).  We summarize the construction of the second term of $p$'s polynomial below:\\

\begin{tabular}{l|l}
Form of $p$'s xor & Form of second term in product-defined polynomial for $p$\\ \hline
$(l_j\oplus l_k)$ & $(y_j+y_k-1)$ \\
$(\neg l_j\oplus l_k)$ & $((1-y_j)+y_k-1)$ \\
$(l_j\oplus \neg l_k)$ & $(y_j+(1-y_k)-1)$ \\
$(\neg l_j\oplus \neg l_k)$ & $((1-y_j)+(1-y_k)-1)$\\
\end{tabular}

\vspace{2mm}
The product of the two described terms gives the polynomial for $p$. Note that $l_j$ and $l_k$ must be different literals from Property 2, but that $i$ might match one of $j$ or $k$.  Since the literals fixed in Stage 1 already have permanently-fixed truth values, the corresponding 0s or 1s are substituted into the predicate polynomials defined above.  
For example, the predicate $\neg l_i \vee (l_j\oplus \neg l_k)$ gives polynomial $(y_i)(y_j+(1-y_k)-1)$ which simplifies to $y_i(y_j-y_k).$
For further example, if $l_i$ was set to true (a.k.a. 1) in Stage 1, and $l_j$ and $l_k$ remain unfixed, then $\neg l_i \vee (l_j\oplus \neg l_k)$ would produce the simple predicate polynomial $y_j-y_k.$

The constructed a set of predicate polynomials has cardinality $|\mathcal{P}|$.
Each predicate polynomial is the product of two terms each of degree at most 1: each predicate polynomial has maximum total degree at most 2. Also, as in the previous reduction, each predicate polynomial contains at most 3 variables (corresponding to a limit of three signed literals per OXR predicate).

Let $\mathcal{F}$ denote the system of polynomials containing both the literal polynomials and the predicate polynomials. By our construction, every satisfying assignment for $(\mathcal{P},\mathcal{L})$ can be interpreted as a mutual root of the polynomials in $\mathcal{F}$. 
Since $(\mathcal{P},\mathcal{L})$ is satisfiable, at least one mutual root exists. Apply algorithm $\mathcal{A}$ to solve the $q$\nobreakdash-Fractional Gr\"{o}bner problem for $\mathcal{F}$ for $\ep=(4/5+\gamma)$ for some fixed $\gamma>0$. Let $Y''$ denote the variables $\mathcal{A}$ selects to ignore: $Y''$ was chosen so that

\vspace{-5mm}

\begin{align*}
|\mathcal{F}_{Y''}|\leq (1-\ep)|\mathcal{F}|\leq (1/5-\gamma)|\mathcal{F}|&= (1/5-\gamma)(|\mathcal{P}|+|\mathcal{L}|)\\
&\leq (1/5-\gamma)\Big(\frac{5}{2}|\mathcal{P}|\Big)\\
&\leq \Big(\frac{1}{2}-\frac{5}{2}\gamma\Big)|\mathcal{P}|.
\end{align*}

\noindent The third line follows from the second line due to Lemma \ref{twofifths}. Let $\mathcal{P}_D$ denote the set of predicate polynomials in $\mathcal{F}_{Y''}$, and $\mathcal{P}_R$ denote the set of predicate polynomials in $\mathcal{F}\backslash\mathcal{F}_{Y''}$. Since $|\mathcal{P}_D|\leq |\mathcal{F}_{Y''}|\leq \Big(\frac{1}{2}-\frac{5}{2}\gamma\Big)|\mathcal{P}|$,

\begin{align}\label{half}
|\mathcal{P}_R|= |\mathcal{P}|- |\mathcal{P}_D| \geq |\mathcal{P}|- \Big(\frac{1}{2}-\frac{5}{2}\gamma\Big)|\mathcal{P}|= \Big(\frac{1}{2}+\frac{5}{2}\gamma\Big)|\mathcal{P}|.
\end{align}

That is, $\mathcal{A}$ computes a Gr\"{o}bner basis with respect to a lexicographic order for the ideal generated by the polynomials in $\mathcal{F}\backslash \mathcal{F}_{Y''}$, and inequality (\ref{half}) says that at least a $(\frac{1}{2}+\frac{5}{2}\gamma)$ fraction of all predicate polynomials must be in $\mathcal{F}\backslash \mathcal{F}_{Y''}$. 
As in the proof of Theorem \ref{extfrachard}, given the Gr\"{o}bner Basis for $\langle\mathcal{F}\backslash \mathcal{F}_{Y''}\rangle$ with respect to a lexicographic order, a mutual root, $y*$, of the polynomials in $\mathcal{F}\backslash \mathcal{F}_{Y''}$ can be efficiently computed via successive elimination of the variables. Each entry in $y*$ corresponds to some literal variable in our Max-OXR instance:
if $y_i$ is 1 in $y*$, assign the corresponding literal $l_i$ to be true, if $y_i$ is 0 in $y*$, assign the corresponding literal $l_i$ to be false. Since $y*\in \{0,1\}^{|Y\backslash Y''|}$, this routine makes an assignment for every $l_i$ corresponding to a  $y_i \in Y\backslash Y''$. 

The vector $y^*$ is a mutual zero of all polynomials in $\mathcal{F}\backslash \mathcal{F}_{Y''}$.
Consider the factored form of the predicate polynomials we constructed: such polynomials evaluate to 0 under our constructed truth assignment if, and only if, either: (a) the signed literal corresponding to $p$'s special position is True or (b) the xor over $p$'s symmetric-position signed literals is True (or both). Thus, 
our $y*$-based partial truth assignment satisfies every predicate whose polynomial is in $\mathcal{P}_R$.

\vspace{2mm}

\noindent \textbf{Stage 3. Supplemental Random Assignment.} Literals corresponding to variables in $Y''$ have not yet been assigned truth values. For these literals, consider an independent coin-flip procedure that assigns True with probability $1/2$, and False with probability $1/2$. We argue that regardless of the partial truth assignment constructed in Stage 2 (and effectively in Stage 1), in expectation, after this procedure, at least half of the predicates corresponding to the polynomials in $\mathcal{P}_D$ are satisfied. Such a procedure can be derandomized via the method of conditional expectations.

Let $p$ denote an OXR predicate corresponding to a polynomial in $\mathcal{P}_D\subseteq \mathcal{F}_{Y''}$. From the form of $\mathcal{F}_{Y''}$,  $p$ contains (some form of) at least one literal corresponding to a $y_i\in Y''$: before the coin-flip procedure, the truth value of such a literal has not yet been decided. Through case analysis, we show that the probability that $p$ is satisfied at the end of the coin-flip procedure is always at least $1/2$.
Suppose that before coin-flipping:
\begin{itemize}
\item The truth value of the signed literal in $p$'s special position has not yet been decided. The coin-flip procedure sets this special-position signed literal to be true with probability $1/2$. Thus, the probability that $p$ is satisfied at the end of the coin-flip procedure is at least $1/2$. 
\item Otherwise, the truth value of the signed literal in $p$'s special position was already decided. Then either:

\vspace{-3mm}

\begin{itemize}
\item Exactly one of the signed literals in a symmetric position of $p$ has not yet been decided. Since the signed literal in $p$'s other symmetric position has a fixed value, the probability that $p$'s xor over the two symmetric positions is true as a result of the coin-flip procedure is $1/2$. Thus, the probability that $p$ is satisfied at the end of the coin-flip procedure is at least $1/2$. 
\item Otherwise, both of the signed literals in the symmetric positions of $p$ have not yet been decided. From Property 1 (proved to hold WLOG in Appendix 2), these symmetric-position signed literals correspond to unique literals. Thus, the probability that $p$'s xor over the two symmetric positions is true as a result of the coin-flip procedure is $1/2$ ($p$'s xor is satisfied by exactly half of the possible truth assignments). Thus, the probability that $p$ is satisfied at the end of the coin-flip procedure is at least $1/2$. 
\end{itemize}
\end{itemize}

\vspace{-3mm}

\noindent Thus, if predicate $p$ corresponds to a polynomial from $ \mathcal{P}_D$, the probability $p$ is satisfied at the end of Stage 3 is at least $1/2$. The expectation of the sum is the sum of the expectations, so  the expected number of predicates from  $\mathcal{P}_D$ that are satisfied after the coin-flip procedure is at least $|\mathcal{P}_D|/2$. Finally, we have a full truth assignment for literals $\mathcal{L}$ that satisfies all predicates in $\mathcal{P}_R$ and at least half of $\mathcal{P}_D$'s predicates, so:

\begin{align*}
\text{Total Predicates We Satisfy}\geq|\mathcal{P}_R|+\frac{|\mathcal{P}_D|}{2} = 
|\mathcal{P}_R|+\frac{(|\mathcal{P}|-|\mathcal{P}_R|)}{2}&= 
\frac{|\mathcal{P}_R|}{2}+\frac{|\mathcal{P}|}{2}.\\
\end{align*}
Using inequality (\ref{half}),

\begin{align*}
\frac{|\mathcal{P}_R|}{2}+\frac{|\mathcal{P}|}{2}&\geq 
\frac{1}{2}\Big(\frac{1}{2}+\frac{5}{2}\gamma\Big)|\mathcal{P}|+\frac{|\mathcal{P}|}{2}\\
&\geq 
\Big(\frac{6}{8}+\frac{5}{4}\gamma\Big)|\mathcal{P}|\\
&\geq 
\Big(\frac{6}{8}+\delta\Big)|\mathcal{P}| \hspace{5mm} \text{ for some  $\delta>0$.}
\end{align*}

\noindent 
Since $\gamma>0$, let $\delta=\frac{5}{4}\gamma$ to get the final statement about $\delta$. This reduction runs in polynomial time. Thus we violate the hardness-of-approximation bound of H{\aa}stad for Max\nobreakdash-OXR (listed as Theorem \ref{hastad2} earlier), obtaining a contradiction.
\end{proof}



\section{Proofs for the Strong $c$\nobreakdash-Partial Gr\"{o}bner problem}
\label{sec:partial}

\begin{proof}[Proof of Theorem \ref{thm:stronger}]
As in the proofs of Theorem \ref{thm:main1} and \ref{thm:main2}, consider a 3-CNF Boolean formula $\Phi$ and define functions $f_C \in \K[x_1,\ldots,x_n]$ for each clause $C$ of $\Phi$. Let $\F$ denote the family of such functions, and, as in De Loera et al.~\cite{deloera}, define $c+1$ copies $\F_1, \F_2, \ldots, \F_{c+1}$ of $\F$ on disjoint sets of variables, such that $\F_i\in \K[x_{i,1}, x_{i, 2}, \ldots, x_{i, n}]$. Then, the system $\bigcup_i \F_i$ has a solution if and only if $\F$ has a solution.

Now, let $g$ denote the polynomial $x + \sum_{i, j} x_{i, j}$, where $x$ is a new variable that does not appear in any other polynomial. Observe that the system $\G := \{g\} \cup \bigcup_i \F_i$ has a solution if and only if $\F$ does, since for each solution of $\bigcup_i \F_i$ there exists a (unique) value of $x$ that yields a solution for $g$.  Moreover, for any subset $\Hh\subset \G$ with $|\Hh| = c$, the system $\G \backslash \Hh$ still has a solution if and only if $\F$ does, because there must be at least one $\F_i$ which is disjoint from $\Hh$. Since $\F$ has a solution exactly when $\Phi$ is satisfiable, we conclude that $\G \backslash \Hh$ has $\{1\}$ as its Gr\"obner basis if and only if $\Phi$ is satisfiable.

Therefore, solving the Strong $c$\nobreakdash-Partial Gr\"obner problem enables us to solve 3SAT with a polynomial-time reduction, using degree-3 polynomials and requiring only that we distinguish the (polynomial-size) Gr\"obner basis $\{1\}$. 
\end{proof}

In our proof of Theorem \ref{prop:lundgroebner} we use a result on graph coloring due to Lund and Yannakakis \cite{lund}.

\begin{thm}[Theorem 2.8 in \cite{lund}]\label{lundy}  For every constant $g>1$, there is a constant $k_g$ such that the following problem is NP-hard. Given a graph $G$, distinguish between the case that $G$ is colorable with $k_g$ colors and the case that the chromatic number of $G$ is at least $gk_g$.
\end{thm}

\begin{proof}[Proof of Theorem \ref{prop:lundgroebner}]
Our proof is by contradiction. For arbitrary fixed $h>0$, assume that for every possible degree bound $k$ there exists a polynomial-time algorithm $\A_k$ to solve the Strong $(hk-1)$\nobreakdash-Partial Gr\"{o}bner problem for polynomial systems of maximum degree $k$.

Let $g=h+1$, so that $g>1$. Applying Theorem \ref{lundy}, we conclude there exists a constant $k_g$ such that it is NP-hard to distinguish whether a graph is $k_g$-colorable or has chromatic number at least $gk_g$. Given a graph $G=(V,E)$ that is guaranteed to be one of these two types, we will present a polynomial-time method to distinguish between the two possible types.

As in \cite{deloera}, write the $k_g$-coloring ideal for $G$, $\mathcal{I}_{k_g}(G)$. This coloring ideal is a polynomial system with $|V|$ variables and $|V|+|E|$ polynomials of maximum degree $k_g$. From our initial assumption, there exists a polynomial-time algorithm $\A_{k_g}$ to solve the Strong $(hk_g-1)$\nobreakdash-Partial Gr\"{o}bner Problem in polynomial systems of maximum degree $k_g$, so apply $\A_{k_g}$ to $\mathcal{I}_{k_g}(G)$.  The algorithm $\A_{k_g}$ returns a Gr\"{o}bner basis $\mathcal{G}$ for $\langle\mathcal{F}_{in}\rangle$ and a list of at most $hk_g-1$ independent sets of variables $X_1,...,X_b$. Let $X_{out}= (\cup_{i=1}^b X_i)$.

Consider the form of $\mathcal{F}_{in}$. Each variable $x\in X_{out}$ corresponds to a node $n$ from $G$. By construction of $\mathcal{I}_{k_g}(G)$, $x$ appears in two types of polynomials. First, $x$ is the sole variable in a \emph{node polynomial} that encodes that $n$ must be assigned some $k_g$th root of unity by any point in the variety of $\mathcal{I}_{k_g}(G)$. Second, for each edge that contains node $n$, variable $x$ appears in an \emph{edge polynomial} that encodes that $n$ and some neighbor must be assigned different $k_g$th roots of unity by any point in the variety of $\mathcal{I}_{k_g}(G)$. When all polynomials that contain variables in $X_{out}$ are removed from $\mathcal{F}$ to get $\mathcal{F}_{in}$ the result is a coloring ideal for a subgraph $G'$ of $G$ obtained by removing $b$ independent sets of nodes and all edges containing these nodes from $G$. In particular, $\langle\mathcal{F}_{in}\rangle$ is the coloring ideal that describes proper colorings of such a subgraph $G'$ by $k_g$ colors, $\mathcal{I}_{k_g}(G')$.

If $\mathcal{G}$ is $\{1\}$, then $\mathcal{I}_{k_g}(G')$ must be empty. This means that no $k_g$\nobreakdash-coloring of $G'$ exists. Then, since $G'$ is a subgraph of $G$, no $k_g$-coloring of $G$ exists. Since $G$ is guaranteed to be either $k_g$\nobreakdash-colorable or have $\chi(G)\geq hk_g$, it must be that $\chi(G)\geq hk_g$. Thus, if $\mathcal{G}$ contains 1, we may conclude that $G$ has chromatic number at least $hk_g$.

Suppose that $\mathcal{G}$ is not $\{1\}$. We argue that this implies the \emph{existence} of a proper coloring of $G$ by strictly less than $gk_g$ colors (so that, in fact, we can conclude $G$ is $k_g$\nobreakdash-colorable). Since $\mathcal{G}$ does not contain 1, the variety of  $\mathcal{I}_{k_g}(G')$ is non-empty. Thus, some coloring of $G'$ by $k_g$ colors must exist. Each of the $b\leq hk_g-1$ independent sets of variables, $X_1,...,X_b$, returned by $\A_{k_g}$ corresponds to an independent set of nodes in $G$. Every node of $G$ is either in such an independent set or in $G'$. Consider coloring each such independent set of nodes by a unique color, and using $k_g$ additional colors to properly color the nodes of $G'$. This would produce a complete proper coloring of $G$ by 
at most
\[
k_g+ hk_g-1= k_g+ (g-1)k_g-1=gk_g-1 \text{ colors.}
\]
Thus, if $\mathcal{G}$ is not $\{1\}$, then there exists some proper coloring of $G$ by $gk_g-1$ colors.
So, it can not be that $\chi(G)\geq gk_g$. Since we were guaranteed that $G$ was either $k_g$\nobreakdash-colorable or had $\chi(G)\geq gk_g$, it must be that $G$ is $k_g$\nobreakdash-colorable. So we conclude that $G$ is $k_g$\nobreakdash-colorable.

For arbitrary $G$, to distinguish between the two possible cases our method constructs a polynomially-sized input to $\A_{k_g}$, which from our initial assumption has polynomial running time. Therefore, in polynomial time our method accurately solves a problem we know to be NP-hard from Theorem \ref{lundy}. Thus (unless P$=$NP) we have a contradiction, and must reject our initial assumption. 
\end{proof}

Our decision rule in this reduction depends only on whether or not the Gr\"{o}bner basis is $\{1\}$. Thus, again, our reduction applies whether  $\A_{k_g}$ is polynomial in the size of our constructed polynomial system, $\mathcal{I}_{k_g}(G)$, or polynomial in the size of the Gr\"{o}bner basis that should be returned. Namely, if $\{1\}$ is the Gr\"{o}bner basis this information will be returned in polynomial time regardless. Otherwise, if $\{1\}$ is not returned in polynomial time it is clear that $\mathcal{G}$ is not $\{1\}$ and we must conclude that $G$ is $k_g$\nobreakdash-colorable.

\section*{Acknowledgments}
\noindent The authors would like to thank David Cox and Jes\'us De Loera for constructive input on these ideas.  D.R.~was partially supported by NSF Grant No.~1122374. G.S. was supported  by NSF Grant No. DMS-1440140 while in residence at the Mathematical Sciences Research Institute in Berkeley, California during the Fall 2017 semester.

\bibliographystyle{plain}
\bibliography{references}

\section*{Appendix 1:  Max Not-2 Instance Preprocessing \\(for Proof of Theorem \ref{extfrachard}, Stage 1)}

Our proof of Theorem \ref{extfrachard} is by contradiction: we aim to construct a polynomial-time algorithm for the satisfiable Max Not-2 Problem with guaranteed approximation performance that contradicts the $(\frac{5}{8}+\delta)$ bound of H{\aa}stad \cite{hast14}. In this Appendix we argue that, given an arbitrary satisfiable input $(\mathcal{P}, \mathcal{L})$ for the satisfiable Max-Not-2 Problem of arity 3, the instance may be preprocessed (removing some predicates and fixing some literals to be true/false) so that several useful properties hold. 
Furthermore, each predicate removal and literal assignment we make to $(\mathcal{P}, \mathcal{L})$ below preserves that a satisfying assignment exists and that if we obtain an $\alpha$-approximate logical assignment for the updated instance, then this logical assignment is \emph{at minimum} $\alpha$-approximate for the original instance. As a result, constructing a polynomial-time algorithm for the Max Not-2 Problem for pre-processed instances (for which convenient special properties hold) with approximation guarantee $(\frac{5}{8}+\delta)$ will be sufficient to obtain the contradiction we seek.

First we remove some predicates and literals from $(\mathcal{P}, \mathcal{L})$. Iterate through the predicates in $\mathcal{P}$ one at a time.  If $p\in \mathcal{P}$ has strictly more than one signed literal corresponding to a single literal $l_i$, then update $(\mathcal{P}, \mathcal{L})$ according to which of the following cases applies:
\begin{enumerate}
\item If $p$ contains three identical signed literals, then since $p$ is a Not-2 predicate, $p$ is trivially satisfied (every logical assignment for $\mathcal{L}$ satisfies $p$). Remove $p$ from $\mathcal{P}$.
\item Otherwise, if $p$ contains exactly 2 identical signed literals indexed by $i$ then:
\begin{enumerate}
\item If $p$'s third signed literal is the other form of $l_i$: every satisfying assignment for $(\mathcal{P}, \mathcal{L})$ must cause $p$ to contain exactly 1 true literal (the only alternative is 2, which can't be satisfying). Thus, we know unequivocally the value $l_i$ must take in every satisfying assignment. Substitute the forced value of $l_i$ into every predicate containing a signed form of literal $l_i$. Say $l_i$ has been permanently fixed. Remove 
$p$ from $\mathcal{P}$. 
\item If $p$ has no third signed literal: every satisfying assignment for $(\mathcal{P}, \mathcal{L})$ must cause $p$ to contain exactly 0 true literals (the only alternative is 2, which can't be satisfying). Thus, we know unequivocally the value $l_i$ must take in every satisfying assignment. Substitute the forced value of $l_i$ into every predicate containing a signed form of literal $l_i$. Say $l_i$ has been permanently fixed. Remove 
$p$ from $\mathcal{P}$. 

\item If $p$'s third signed literal corresponds to some other index $j$ where $l_j$ has not yet been fixed, then do nothing.

\item If $p$'s third signed literal originally corresponded to some other index $j$ where $l_j$ has already been permanently fixed to 0 or 1, then do nothing. 

\end{enumerate}
\item Otherwise, it must be that $p$ contains 2 signed literals corresponding to the same index $i$, but with opposing signs (one negated form and one positive form).  Then:
\begin{enumerate}

\item If $p$'s third signed literal is also for index $i$, then it must be identical to one of $p$'s opposing-signed literals: this was already covered in sub-case 2(a) above.
\item  If $p$'s third signed literal corresponds to some other index $j$ for which $l_j$ has not yet been fixed: since exactly one of $p$'s opposing-sign $i$ literals are true, we know unequivocally the value $l_j$ must take in every satisfying assignment. Substitute the forced value of $l_j$ into every predicate containing a signed form of literal $l_j$. Say $l_j$ has been permanently fixed. 
 Remove 
 $p$ from $\mathcal{P}$. 
\item If $p$'s original third signed literal corresponded to some other index $j$ for which $l_j$ has already been permanently fixed to 0 or 1: since variables are only fixed to values we know they must take in every satisfying assignment, $l_j$ must have been fixed so that the signed literal in predicate $p$ corresponding to $l_j$ was false.\footnote{Otherwise, the signed literal corresponding to $l_j$ would be true, and since exactly one of $p$'s opposing-sign literals involving $l_i$ is true, no choice of $l_i$ would satisfy $p$, and this would contradict the satisfiability of $(\mathcal{P}, \mathcal{L})$. } Thus, $p$ will be satisfied by any assignment for $l_i$. Remove $p$ from $\mathcal{P}$.

\item If $p$ never had a third signed literal: since $p$ is a Not-2 predicate, $p$ is trivially satisfied (every assignment for $\mathcal{L}$ satisfies $p$). Remove $p$ from $\mathcal{P}$.

\end{enumerate}
\end{enumerate}

Call the set of all literals fixed during this procedure $\mathcal{L}^f$, and the set of all predicates removed from the original $\mathcal{P}$ by $\mathcal{P}^r$.
After executing the above procedure for every $p\in \mathcal{P}$, observe that $(\mathcal{P}, \mathcal{L})$ has the following property.  \\

\noindent \textbf{Property 1.} \textit{Any predicate $p'\in \mathcal{P}$ that has multiple occurrences of the same literal must have a very specific form: either $p'$ contains two identical signed literals and a third signed literal corresponding to a different index, or  $p'$ contains two identical signed literals and a third literal whose truth value has been permanently fixed to either false or true. These forms correspond to sub-cases 2(c) and 2(d) above: in all other sub-cases, $p$ was removed from $\mathcal{P}$.}\\

Further observe that truth values were only permanently fixed when we could reason unequivocally about the value they must take in every satisfying assignment. Thus, since the original $(\mathcal{P}, \mathcal{L})$ was satisfiable, the updated $(\mathcal{P}, \mathcal{L})$ with the current partial logical assignment for $\mathcal{L}^f$ is still satisfiable. Before proceeding, notice also that a predicate was only removed from $\mathcal{P}$ when we could be certain that it would be satisfied by an arbitrary logical assignment that extends the partial assignment already constructed for permanently-fixed variables in $\mathcal{L}^f$. 

Next, we make a final update to  $(\mathcal{P}, \mathcal{L})$. Call any literal $l\in \mathcal{L}$ which appears in at most one predicate from $\mathcal{P}$ a \textit{loner literal}. Call the set of predicates from $\mathcal{P}$ which contain at least one loner literal by $\mathcal{P}^l$. 
We consider temporarily ignoring predicates in $\mathcal{P}^l$ until after all non-loner literals have been fixed. 

Consider an arbitrary partial assignment for $(l_1,...,l_{|\mathcal{L}|})$ that fixes every literal \textit{except for the Loner Literals}. If a Not-2 predicate  $p\in \mathcal{P}$ contains one or more loner literals (in either positive or negated form), then this arbitrary partial assignment may be easily extended to satisfy $p$: by manipulating the $\{T,F\}$ value of a contained loner literal at least two distinct total numbers of true signed literals can be reached for predicate $p$ (this follows from Property 1). At most one of these totals can be equal to 2, and the other must be some acceptable total for $p$ (such that $p$ is satisfied). Since the loner literals in $p$ appear in no other predicates (by definition), their value may be fixed one-by-one in this way to satisfy all predicates in $\mathcal{P}^l$. Since we can easily completely satisfy predicates containing loner literals at the end, we ignore such predicates for now: if a predicate contains a loner literal, remove that predicate from $\mathcal{P}$. Next remove all loner literals from $\mathcal{L}$. 

Notice that removing the predicates in $\mathcal{P}^l$ may have caused some additional literals to become loner literals. Successively remove additional rounds of loner-containing predicates and loner literals. Mark each loner literal by the round in which it was removed: once we have created a partial assignment for the remaining system we will fix the values of the loner literals in an order that reverses the order in which they were removed from $\mathcal{L}$. Such an ordering ensures that as loners from each round are returned to the instance, the argument made in the previous paragraph about how to choose their value will always apply. 

We now have the $(\mathcal{P}, \mathcal{L})$ that we will argue about for the remainder of the reduction.  We make a few observations before starting Stage 2. Property 1 still holds, and as a result of the loner literal-removal and loner-containing-predicate removal process, we have:\\

\noindent \textbf{Property 2.} \textit{After the updates in Stage 1, Every literal $l\in \mathcal{L}$ appears in some form (negated or positive) in at least two predicates from $\mathcal{P}$.}\\

Consider the two types of predicates removed from the original instance during Stage 1.  Each predicate in $\mathcal{P}^r$ (those removed in the first part of Stage 1) is already guaranteed to be satisfied by any extension of the partial assignment that has been  permanently fixed on $\mathcal{L}^f$. Further, for any partial assignment for literals now remaining in $\mathcal{L}$ that leaves the values of loner literals (from every round) unassigned, 
each predicate removed for containing a loner literal (in any round) can be efficiently satisfied by appropriate choices for the loner literals. 

Since $100\%$ of the predicates removed from $\mathcal{P}$ in Stage 1 can be satisfied in one of these two ways (in polynomial time), any fraction of the predicates we can satisfy for the remaining updated instance $(\mathcal{P}, \mathcal{L})$ will be a lower bound on the fraction of the original predicates that are satisfied. Thus, to get a contradiction, it is sufficient to show that for our remaining satisfiable Max-Not-2 instance we can satisfy a $(\frac{5}{8}+\epsilon)$ fraction of the remaining predicates $\mathcal{P}$. We construct such an assignment for the remaining literals in $\mathcal{L}$ over stages 2 and 3.\\


\section{Appendix 2: OXR Instance Preprocessing \\(for Proof of Theorem \ref{extfracharddeg2}, Stage 1) }

Conceptually the following appendix provides arguments highly analogous to those in Appendix 1. Details vary from Appendix 1 because of the different form of the predicates (OXR rather than Not-2 predicates) so we include them for completeness.

Suppose that we are given an arbitrary satisfiable input $(\mathcal{P}, \mathcal{L})$ for the \textit{Max-OXR Problem of arity 3}. We argue that the instance may be preprocessed (removing some predicates and fixing some literals to be true/false) so that several useful properties hold.
Each predicate removal and literal assignment we make to $(\mathcal{P}, \mathcal{L})$ preserves the properties that a satisfying assignment exists and that if we obtain an $\alpha$-approximate logical assignment for the updated instance, then this logical assignment is \emph{at minimum} $\alpha$-approximate for the original instance.

For an OXR predicate of form:
\vspace{-2mm}
\begin{align*}\label{oxrform}
OXR(q_1,q_2,q_3)=q_1 \vee (q_2 \oplus q_3),
\end{align*}
we say that $q_1$ is in the \textit{special position} of $p$ and that $q_2$ and $q_3$ are in the \textit{symmetric positions} of $p$. Iterate through the predicates in $\mathcal{P}$.  Consider the signed literals in the symmetric positions of $p$: if both of these signed literals correspond to the same literal, then update $(\mathcal{P}, \mathcal{L})$ according to which of the following cases applies.

\begin{enumerate}
\item If the signed literals in the symmetric positions of $p$ are identical, then their xor must be false (either both of the symmetric-position signed literals are true, or both of the symmetric-position signed literals are false). Thus, in every satisfying assignment the special-position signed literal of $p$ must be true. 
Let $l_i$ denote the literal corresponding to the signed literal in the special position of $p$. Substitute the forced value of $l_i$ into every predicate containing a signed form of $l_i$. Say $l_i$ has been permanently fixed. Remove 
$p$ from $\mathcal{P}$. 

\item Otherwise the signed literals in the symmetric positions of $p$ are in opposing forms (one positive, one negated). In this case their xor is true for every possible assignment. Since $p$ is trivially satisfied, remove $p$ from $\mathcal{P}$.
\end{enumerate}
    
   Call the set of all literals fixed during this procedure $\mathcal{L}^f$, and the set of all predicates removed from the original $\mathcal{P}$ by $\mathcal{P}^r$.
After executing the above procedure for every $p\in \mathcal{P}$, observe that $(\mathcal{P}, \mathcal{L})$ now has the following property.  \\

\noindent \textbf{Property 1.} \textit{If $p\in \mathcal{P}$, then the two signed literals in the symmetric positions of $p$ correspond to unique literals.}\\

Literals were only permanently fixed (and removed) when we could reason unequivocally about the truth value they must take in every satisfying assignment. Thus, since the original $(\mathcal{P}, \mathcal{L})$ was satisfiable, the updated $(\mathcal{P}, \mathcal{L})$ is still satisfiable. A predicate was only removed from $\mathcal{P}$ when we could be certain that it would be satisfied by any assignment that extends the partial assignment already fixed for $\mathcal{L}^f$.
    
As in the previous proof, we make a final update to  $(\mathcal{P}, \mathcal{L})$ to remove literals that appear in only one predicate (and the predicates that contain such literals). Call any literal $l\in \mathcal{L}$ which appears in at most one predicate from $\mathcal{P}$ a \textit{loner literal}. Call the set of predicates from $\mathcal{P}$ which contain a loner literal by $\mathcal{P}^l$.  We consider temporarily ignoring predicates in $\mathcal{P}^l$ until all non-loner literals have been fixed. For $p\in \mathcal{P}^l$, suppose that the truth values for all literals in $p$, except one loner-literal $l_i$, have already been fixed.\footnote{If more than one signed loner literal in $p$ remains unfixed, then fix all but one arbitrarily, then proceed.} Then:
\begin{itemize}
\item If $l_i$ corresponds to the signed literal in the special position of $p$, then clearly $l_i$ may be set so $p$'s special-position signed literal is true (and $p$ is satisfied). By definition, this choice for $l_i$ (a loner literal) affects no other predicates.
\item If $l_i$ corresponds to a signed literal in a symmetric position of $p$, then, regardless of the truth value of the signed literal in $p$'s other symmetric position\footnote{From Property 1, this other signed literal is not a form of $l_i$.}, there is a truth assignment for $l_i$ that makes $p$'s xor true (so that $p$ is satisfied). By definition, this choice for $l_i$ affects no other predicates. 
\end{itemize}

Thus, we ignore OXR predicates containing loner literals for now: given an arbitrary partial assignment for $(l_1,...,l_{|\mathcal{L}|})$ that fixes every literal \textit{except for the loner literals}, a predicate  $p\in \mathcal{P}$ that contains one or more loner literals can be satisfied by manipulating the $\{T,F\}$ value of the contained loner literal. 
Thus, if a predicate contains a loner literal, remove that predicate from $\mathcal{P}$. Next remove all loner literals from $\mathcal{L}$. 

Removing predicates may cause additional literals to become loner literals. Successively remove additional rounds of loner-containing predicates and loner literals. Mark each loner literal by the round in which it was removed: once we have created a partial assignment for the remaining system we will fix the values of the loner literals in an order that reverses the order in which they were removed from $\mathcal{L}$ such that all loner-containing predicates are satisfied. 

We now have the $(\mathcal{P}, \mathcal{L})$ that we will argue about for the remainder of the reduction. As in the previous proof, observe that: \\

\noindent \textbf{Property 2.} \textit{After the updates in Stage 1, Every literal $l\in \mathcal{L}$ appears in some form (negated or positive) in at least two predicates from $\mathcal{P}$.}\\

As in the proof of Theorem \ref{extfrachard}, since $100\%$ of the predicates removed from $\mathcal{P}$ in Stage 1 can be polynomial-time satisfied after any partial assignment is fixed for the remaining instance $(\mathcal{P}, \mathcal{L})$, it is sufficient to show that for this remaining instance we can satisfy a $(\frac{6}{8}+\epsilon)$ fraction of the predicates in polynomial time. We construct such an assignment for the remaining literals over stages 2 and 3.
\end{document}